\documentclass[12pt]{article}
\usepackage[margin=1.2in]{geometry}

\usepackage{url,amsmath,amssymb,amsthm,natbib,appendix,bbm,mathtools,booktabs,caption,makecell,float,enumerate,scalerel,enumitem}
\usepackage[dvipsnames]{xcolor}
\usepackage{pgf,tikz,pgfplots,mathrsfs,graphicx}
\usepackage{hyperref,cleveref}
\usetikzlibrary{arrows}
\theoremstyle{plain}
\newtheorem{definition}{Definition}
\newtheorem{thm}{Theorem}
\newtheorem{prop}{Proposition}

\theoremstyle{definition}

\theoremstyle{remark}

\usepackage{chngcntr,apptools}
\AtAppendix{\counterwithin{lem}{section}}
\AtAppendix{\counterwithin{prop}{section}}
\AtAppendix{\counterwithin{axm}{section}}

\def\R{\mathbb{R}}
\def\D{\mathcal{D}}
\def\E{\mathbb{E}}

\title{On the (Non-)Uniqueness of Random Non-Expected Utility}
\author{Yi-Hsuan Lin\thanks{Institute of Economics, Academia Sinica; \href{mailto:yihsuanl@econ.sinica.edu.tw}{yihsuanl@econ.sinica.edu.tw}. I am deeply grateful to Larry Epstein for his advice and encouragement.}}
\date{July 2026}
\begin{document}
\maketitle\begin{NoHyper}
\abstract{In random expected utility \citep*{Gul:2006}, the distribution of preferences is uniquely identified from random choice. This paper investigates whether such identification extends beyond expected utility. We first show that when risk preferences conform to the disappointment aversion model of \cite{Gul:1991}, the distribution of preferences remains uniquely identified. To assess the scope of this result, we then examine other models of non-expected utility. Within the broader class of betweenness preferences \citep*{Dekel:1986}, random utility can be unidentifiable. If preferences are confined to the weighted expected utility class \citep{Chew:1983}, a more nuanced picture emerges: unique identification holds in a three-prize setting but fails with four or more prizes. These findings show that the uniqueness property of random expected utility may persist beyond expected utility, but its persistence critically depends on the class of risk preferences under consideration.}

\par\ \\
{\bf Keywords:} random choice, random utility, random preference, identification, non-expected utility

\section{Introduction}
A fundamental question in decision theory is whether choice behavior uniquely identifies the primitives of a behavioral model. In deterministic settings, this question concerns the recovery of a single preference relation and, when applicable, the uniqueness of its representation. In random choice environments, however, the object of interest is no longer a single preference but a distribution over preferences. Random utility theory provides a natural framework for random choice. Under the individual interpretation, choices appear random because a decision maker's preference varies with an unobserved state. Under the population interpretation, choices appear random because of heterogeneity in preferences across individuals. In either interpretation, random choice is generated by a distribution over preferences, and identification asks whether this distribution is uniquely determined by the observed choice frequencies. In general, the answer is negative: distinct distributions over preferences may generate identical random choice behavior \citep{Fishburn1998}.
\par
A notable exception is the random expected utility model, where choice alternatives are lotteries and all preferences conform to expected utility theory. Under this restriction, the distribution of preferences is uniquely recoverable from random choice \citep{Gul:2006}. The focus on expected utility is natural as a first step, but is not completely satisfactory in light of its well-known descriptive failures, such as the Allais paradox. Experimental evidence suggests that individuals not only differ in their risk attitudes but also in the extent to which they depart from expected utility. If random choice reflects such heterogeneity, a natural question is whether the underlying distribution of non-expected utility preferences can also be uniquely identified.
\par
We address this question in the context of lotteries over a finite set of monetary prizes and within the class of betweenness preferences introduced by \cite{Dekel:1986}. Two notable subclasses are disappointment aversion preferences \citep{Gul:1991} and weighted expected utility \citep{Chew:1983}. Betweenness preferences generalize expected utility by retaining linear indifference curves while allowing them to be non-parallel. This structure permits accommodation of Allais-type behavior without sacrificing much of the analytical tractability of expected utility. Betweenness preferences therefore provide a natural framework for studying identification beyond expected utility.
\par
Our findings provide both positive and negative answers to this question. First, when preferences are confined to the disappointment aversion model, the distribution over preferences remains uniquely identified from random choice. Disappointment aversion preferences are arguably a minimal departure from expected utility, introducing only one additional parameter to capture attitudes toward disappointment and elation. The model nevertheless retains much of the analytical structure of expected utility. Our result shows that this additional flexibility does not come at the cost of identification: random choice and preference heterogeneity remain in one-to-one correspondence.
\par
This positive result, however, does not extend to the broader class of betweenness preferences. We show by example that unique identification fails when preferences are restricted to the betweenness class and are monotone with respect to first-order stochastic dominance. Non-uniqueness even obtains within the class of weighted expected utility, another parsimonious subclass of betweenness that admits a finite-dimensional parameterization. Interestingly, unique identification still holds for weighted expected utility in a three-prize setting, but fails once four or more prizes are available.
\par
Taken together, two broader lessons emerge from these findings. On the positive side, unique identification of the distribution over preferences survives in an economically important class of non-expected utility. Random choice can fully reveal the distribution of a decision maker's latent attitudes toward both risk and disappointment, and also pin down the heterogeneity of such attitudes across individuals. On the negative side, the boundary between uniqueness and non-uniqueness is surprisingly subtle. The weighted expected utility example demonstrates that neither finite-dimensional parameterization nor axiomatic and geometric closeness to expected utility is sufficient to guarantee unique identification of random non-expected utility. Understanding what gives rise to uniqueness remains an important direction for future research.
\par
The rest of the paper is organized as follows. The remainder of Section 1 reviews the related literature. Section 2 introduces the framework for random choice and random utility and defines the classes of risk preferences under study. Section 3 presents the main identification results for random non-expected utility. The appendix contains most of the proofs.  

\paragraph{Literature Review}
\ \par\ 
Since the pioneering work of \cite{BlockMarschak1960}, random utility theory has served as a standard rational model of stochastic choice. However, as demonstrated by \cite{Fishburn1998}, a well-known difficulty in this framework is that the underlying distribution over preferences is generally not uniquely identified from stochastic choice data. Uniqueness may nevertheless be restored when the structure of the choice environment permits economically meaningful restrictions on preferences. For example, \cite{single-crossing-RUM} considers an ordered space of choice alternatives and assumes that preferences satisfy single-crossing. \cite{YANG2023} considers consumption goods with monetary costs and restricts preferences to be quasi-linear. In the context of choice under risk, \cite{Gul:2006} establish unique identification when preferences are restricted to expected utility. 
\par
Similar identification results have also been established for richer models of decision under uncertainty. In the Anscombe-Aumann framework, \cite{Lu:2016} shows that random subjective expected utility is uniquely identified from stochastic choice data. \cite{Lu:2020} further extends this uniqueness result to random maxmin expected utility. While maxmin expected utility is introduced to accommodate ambiguity-sensitive behavior highlighted by the Ellsberg paradox, many departures from expected utility are motivated by violations of the independence axiom, as illustrated by the Allais paradox. Prominent examples include weighted expected utility \citep{Chew:1983}, betweenness \citep{Dekel:1986}, and disappointment aversion \citep{Gul:1991}. The present paper investigates whether similar identification results for random utility extend when risk preferences are restricted to these non-expected utility models.
\par
In purely abstract choice settings where preference restrictions are less natural, alternative approaches to the identification problem have been developed. \cite{SULEYMANOV2024} considers a subclass of random utility representations termed branching-independent and shows that every random choice function consistent with random utility can be generated by a unique branching-independent distribution over preferences. \cite{TURANSICK2022} provides necessary and sufficient conditions on random choice data for which the random utility representation is unique. Finally, unique identification may obtain when choice data is further enriched. For example, \cite{DuttaMasatliogluMu2026} adopt novel choice data, called random binary choice paths, and show that it is rich enough to uniquely identify the underlying preference distribution.

\section{The Model}
This paper focuses on random choice under risk. There is a finite set of monetary prizes denoted $X=\{x_0,x_1,\cdots,x_N\}\subseteq\R$, where $N\geq2$ and $x_0<x_1<\cdots<x_N$. The objects of choice are lotteries over $X$. A lottery $p$ is identified with the vector $(p(x_1),\cdots,p(x_N))\in\R^N$, where for each $n=1,\cdots,N$, $p(x_n)$ denotes the probability of receiving prize $x_n$, and $p(x_0)\coloneqq1-\sum_{n=1}^Np(x_n)$ denotes the probability of receiving prize $x_0$. Where the context is clear, we write $p_n$ for $p(x_n)$.  The set of all lotteries, denoted by $\Delta$, is identified with the probability simplex in $\R^N$: $\Delta\coloneqq\{p\in\mathbb{R}_+^N:\sum_{n=1}^Np_n\leq 1\}$. Finally, with a slight abuse of notation, the degenerate lottery that assigns probability one to prize $x\in X$ is also denoted by $x$.
\par
A choice set, also called a menu, is a finite set of lotteries. Let $\mathcal{D}$ denote the collection of all menus. Choice from a menu is modeled as a random set. Specifically, for any $A,D\in\mathcal{D}$, let $\rho(A,D)$ denote the probability that $A$ is the set of all optimal alternatives from menu $D$. To ensure feasibility, we require that $\rho(A,D)>0$ only if $A$ is a nonempty subset of $D$. The observable choice behavior is summarized by a random choice correspondence, defined as follows.\footnote{We consider multi-valued random choice to avoid dealing with ties. The literature often studies single-valued random choice, where choice from a menu $D$ is summarized by a probability distribution over $D$, rather than over the collection of subsets of $D$. To accommodate single-valued random choice, a random utility model requires either that ties occur with probability zero, or that a tie-breaking rule be specified. Otherwise, the connection between the behavior and the model is not tight. Multi-valued random choice avoids these issues by treating the entire set of utility-maximizing alternatives as the observable outcome.}
\begin{definition}
A random choice correspondence (RCC) is a function $\rho:\mathcal{D}\times\mathcal{D}\to[0,1]$ such that (i) for all $A,D\in\D$, $\rho(A,D)>0$ implies $\emptyset\neq A\subseteq D$, and (ii) for all $D\in\mathcal{D}$, $\sum_{A\subseteq D}\rho(A,D)=1$. 
\end{definition}
\par
A preference $\succsim$ over lotteries is a complete and transitive binary relation on $\Delta$. As usual, let $\succ$ and $\sim$ denote the strict part and the symmetric part of $\succsim$, respectively. Since prizes are monetary, we assume throughout that preferences are monotone with respect to first-order stochastic dominance.\footnote{That is, $p\succsim q$ whenever $p$ first-order stochastically dominates $q$, and $p\succ q$ whenever the dominance is strict. The monetary setting provides a natural justification for imposing this monotonicity, which is used in the proofs of our positive identification results.}
\par
Given a set $\mathcal{P}$ of preferences, a random preference $\pi$ is a probability distribution over $\mathcal{P}$. We say that an RCC $\rho$ is rationalized by random preference $\pi$ if for all $A\subseteq D\in\D$,
$$\rho(A,D)=\pi\bigl\{\succsim\in\mathcal{P}:A=\{p\in D:p\succsim q\ \forall\ q\in D\}\bigr\}.$$
A function $V:\Delta\to\R$ represents preference $\succsim$ if $p\succsim q\Leftrightarrow V(p)\geq V(q)$ for all $p,q\in\Delta$. If every $\succsim\in\mathcal{P}$ is represented by a function $V(\cdot|\alpha)$ parameterized by $\alpha$, where distinct values of $\alpha$ correspond to distinct preferences, then a random preference may equivalently be identified with a distribution over the parameter space.
\par
The core question addressed in this paper is one of identification: can a random choice correspondence be rationalized by more than one random preference? The answer depends crucially on the collection of preferences over which random preferences are defined. If $\mathcal P$ is unrestricted, then, by an argument similar to the classic non-identification example in \citet{Fishburn1998}, one can construct distinct random preferences that rationalize the same random choice correspondence. In contrast, if $\mathcal P$ consists only of expected utility preferences, then random choice uniquely identifies random preference, as shown by \citet{Gul:2006}. This paper studies whether such uniqueness continues to hold for classes of non-expected utility preferences that accommodate classic Allais-type violations of expected utility.

\paragraph{Betweenness preferences and their subclasses}
\ \par\ 
A classic model of choice under risk is expected utility. A preference $\succsim$ is an expected utility (EU) preference if there exists a utility function $u:X\to\R$ such that $V(p)=\sum_{x\in X}u(x)p(x)$ represents $\succeq$. Where the context is clear, we write $u_n$ for $u(x_n)$. By monotonicity of $\succeq$, $u$ is strictly increasing. Moreover, $u$ is unique up to positive affine transformations and thus can be normalized so that $u(x_0)=0$ and $u(x_N)=1$. Under this normalization, $u$ is unique. 
\par
Since lotteries are identified with points in $\R^N$, the indifference sets of an EU preference are parallel hyperplanes. A well-known generalization of expected utility is betweenness preference \citep{Dekel:1986}, which preserves the linearity of indifference sets but relaxes their parallelism.
\begin{definition}
    A preference $\succsim$ is a betweenness preference if there exists $u:X\times[0,1]\to\mathbb{R}$, continuous in the second argument, such that $p\succsim q\Leftrightarrow V(p)\geq V(q)$, where $V(p)$ is defined implicitly as the unique $v\in[0,1]$ that solves 
    $$\sum_{n=0}^Nu(x_n,v)p(x_n)=vu(x_N,v)+(1-v)u(x_0,v).$$
\end{definition}
This representation is often called implicit expected utility. The representation parameter $u(x,v)$ is unique up to positive affine transformations that are continuous in $v$. Hence, $u$ can be normalized so that $u(x_0,v)=0$ and $u(x_N,v)=1$ for all $v\in[0,1]$. Under this normalization, the indifference set associated with utility level $v$ is given by $\sum_{n=1}^{N}u(x_n,v)p_n=v$, which is a hyperplane with normal vector $(u(x_1,v),\cdots,u(x_{N-1},v),1)$. Moreover, since preferences are monotone with respect to first-order stochastic dominance, $u(x,v)$ is weakly increasing in $x$:
$$0\leq u(x_1,v)\leq\cdots\leq u(x_{N-1},v)\leq1$$
for all $v\in[0,1]$.
\par
Expected utility can now be viewed as a special case of betweenness preference, in which $u(x,v)$ is independent of $v$. We next introduce two notable subclasses of betweenness preference: disappointment aversion \citep{Gul:1991} and weighted expected utility \citep{Chew:1983}, both of which also nest expected utility.
\begin{definition}
    A preference $\succsim$ is a disappointment aversion (DA) preference if there exist $u:X\to\mathbb{R}$ and $\beta\in(-1,\infty)$ such that $p\succsim q\Leftrightarrow V(p)\geq V(q)$, where $V(p)$ is defined implicitly as the unique $v\in\R$ that solves 
    $$\sum_{x\in X}\phi(x,v)p(x)=v$$
    with
    $$\phi(x,v)=\begin{cases}u(x)\ &\text{if } u(x)\leq v,\\
    \frac{u(x)+\beta v}{1+\beta}\ &\text{if } u(x)>v.\end{cases}$$
\end{definition}
In a DA representation, $\beta$ is unique, whereas $u$ is unique up to positive affine transformations. Hence, $u$ can be normalized so that $u(x_0)=0$ and $u(x_N)=1$.
\par
\begin{definition}
    A preference $\succsim$ is a weighted expected utility (WEU) preference if there exist $u:X\to\mathbb{R}$ and $w:X\to\R_{++}$ such that the function $V:\Delta\to\R$ defined by
    $$V(p)=\frac{\sum_{x\in X}u(x)w(x)p(x)}{\sum_{x\in X}w(x)p(x)}$$
    represents $\succsim$.
\end{definition}
In a WEU representation, the parameters $u(\cdot)$ and $w(\cdot)$ are unique under the normalization $u(x_0)=0$, $u(x_N)=1$, and $w(x_0)=w(x_N)=1$. Moreover, under this normalization, $V(p)$ is the unique $v\in[0,1]$ that solves
$$\sum_{n=1}^N\left[w(x_n)u(x_n)+(1-w(x_n))v\right]p(x_n)=v.$$
This corresponds to an implicit expected utility representation with parameter
$$\tilde{u}(x,v)=w(x)u(x)+(1-w(x))v.$$
Thus, every WEU preference is a betweenness preference.


\section{The Results}
We begin by reviewing the identification result for random expected utility. The following theorem is due to \citet{Gul:2006}.
\begin{thm}\label{thm:EU}
    A random choice correspondence is rationalized by at most one random expected utility preference.
\end{thm}
\begin{proof}
Under normalization $u(x_0)=0$ and $u(x_N)=1$, the EU representation parameter $u:X\to\R$ is unique. Hence, a random EU preference can be represented by a joint distribution $\mu$ of $(u_1,\cdots,u_{N-1})\in\R^{N-1}$.\footnote{In \cite{Gul:2006}, prizes are abstract objects, and monotonicity is not imposed. In our setting of monetary prizes, monotonicity permits the two-point normalization $u(x_0)=0$ and $u(x_N)=1$. This normalization leads to the following proof based on the Cram\'{e}r-Wold theorem (see, e.g., \citet[Theorem 29.4]{Billingsley1995}), which differs from the argument in \cite{Gul:2006}.}
\par
Fix an interior $p^0\in\Delta$. For every $r\in\mathbb{R}^N$, $p^0+\epsilon r\in\Delta$ when $\epsilon>0$ is sufficiently small. Moreover,
$$\sum_{n=1}^Nu_np^0_n>\sum_{n=1}^Nu_n(p^0_n+\epsilon r_n)\iff\sum_{n=1}^Nu_nr_n<0.$$
Consider $r=(r_1,\cdots,r_{N-1},-t)$. Since $u_N=1$, we have
$$\sum_{n=1}^Nu_nr_n<0\iff\sum_{n=1}^{N-1}u_nr_n< t.$$
If an RCC $\rho$ is rationalized by $\mu$, then
$$\rho(p^0,\{p^0,p^0+\epsilon r\})=\mu\left\{(u_1,\cdots,u_{N-1}):\sum_{n=1}^{N-1}u_nr_n< t\right\}.$$
This holds for all $(r_1,\cdots,r_{N-1})\in\mathbb{R}^{N-1}$ and all $t\in\mathbb{R}$. In other words, we can recover the distribution of every linear combination of $(u_1,\cdots,u_{N-1})$ from random choice. By the Cram\'{e}r-Wold theorem, the joint distribution of $(u_1,\cdots,u_{N-1})$ is uniquely determined. Hence, $\rho$ cannot be rationalized by more than one random EU preference.
\end{proof}
\par
According to the proof, observing how often a fixed interior lottery $p$ is chosen in binary comparisons with other lotteries $q$ is sufficient for identifying the distribution of EU preferences. From a geometric perspective, the proof suggests that we identify the distribution of the normal vector of the indifference set containing $p$ from $\{\rho(p,\{p,q\})\}_{q\in\Delta}$. Then, since all the indifference sets of an EU preference are parallel, we obtain the distribution of EU preferences. When preferences are relaxed to the class of betweenness preferences, indifference sets remain hyperplanes. Hence, we can apply the same technique to identify, for each interior $p$, the distribution of the normal vector of the indifference set containing $p$. The extension to boundary lotteries follows by approximation from interior lotteries.\footnote{Under the implicit expected utility representation with a normalized $u(x,v)$, the normal vector of the indifference set containing lottery $p$ is given by $(u(x_1,V(p)),\cdots,u(x_{N-1},V(p)),1)$. Since $u(x,v)$ is continuous in $v$ and $V(p)$ is continuous in $p$, this normal vector is continuous in $p$. Therefore, if $p^m$ converges to $p$, the distribution of $(u(x_1,V(p^m)),\cdots,u(x_{N-1},V(p^m)),1)$ converges to the distribution of $(u(x_1,V(p)),\cdots,u(x_{N-1},V(p)),1)$.} The remaining challenge is whether we can recover the joint distributions of these normal vectors across different lotteries. 
\par
We then present the main, positive identification result of this paper. It shows that the uniqueness of random EU preference extends to random disappointment aversion. Note that, under normalization $u_0=0$ and $u_N=1$, a random DA preference can also be represented by a joint distribution of the representation parameter $(u_1,\cdots,u_{N-1},\beta)$.
\begin{thm}\label{thm:DA}
    A random choice correspondence is rationalized by at most one random disappointment aversion preference.
\end{thm}
\begin{proof}[Sketch of Proof]
The proof utilizes a special structure of a DA representation for lotteries over the best and second-best prizes. Fix a lottery $p=a x_N+(1-a)x_{N-1}$. For any lottery $q$ equally good as $p$, the utility value of $q$ must lie between $u_{N-1}$ and 1. Thus, prize $x_N$ will be the only elation outcome, and all other prizes are disappointment. Therefore, on the indifference set of $p$, the tradeoff between the probabilities of $x_N$ and $x_{N-1}$ is distorted by $\beta$, but the tradeoff between two prizes worse than $x_N$ is governed only by the utility index $u$.
\par
Such tradeoffs are reflected in the implicit utility index $\phi(x,v)$ at $v=V(p)$. Specifically, 
$$\left(\phi(x_1,V(p)),\cdots,\phi(x_N,V(p))\right)=\left(u_1,\cdots,u_{N-1},\frac{1+\beta(1-a)u_{N-1}}{1+\beta(1-a)}\right).$$
Consider the local probability substitution rate between $x_n$ and $x_{N-1}$ for $n\neq N-1$:
$$\left(\frac{\phi(x_n,V(p))}{\phi(x_{N-1},V(p))}\right)_{n\neq N-1}=\left(\frac{u_1}{u_{N-1}},\cdots,\frac{u_{N-2}}{u_{N-1}},\frac{\phi(x_N,V(p))}{\phi(x_{N-1},V(p))}\right).$$
This is a random vector in $\R^{N-1}$ whose distribution is identifiable from random choice according to the proof of Theorem \ref{thm:EU} and the subsequent discussion. The last entry $\frac{\phi(x_N,V(p))}{\phi(x_{N-1},V(p))}$ depends only on $u_{N-1}$, $\beta$, and $a$ in a seemingly nonlinear way. However, a suitable transformation gives us an affine structure:
$$\frac{\phi(x_{N-1},V(p))}{\phi(x_N,V(p))-\phi(x_{N-1},V(p))}=a\times\frac{u_{N-1}}{1-u_{N-1}}+(1-a)\times\frac{(1+\beta)u_{N-1}}{1-u_{N-1}}.$$
This makes a Laplace-transform argument for recovering the joint distribution applicable. The joint distribution of $\frac{u_1}{u_{N-1}},\cdots,\frac{u_{N-2}}{u_{N-1}}$, $\frac{u_{N-1}}{1-u_{N-1}}$, and $\frac{(1+\beta)u_{N-1}}{1-u_{N-1}}$ is uniquely recoverable from random choice. We can then trace back the distribution of $(u_1,\cdots,u_{N-1},\beta)$. See Appendix for the details.
\end{proof}
\par
\Cref{thm:DA} shows that the identification property of random expected utility is robust to the introduction of disappointment aversion. While DA preferences incorporate an additional parameter to EU representation to accommodate Allais-type departures from expected utility, such flexibility comes at no cost to identification: random choice remains fully informative about preference heterogeneity. In particular, heterogeneity in attitudes toward both risk and disappointment can be uniquely identified from random choice data.
\par
The next result shows that this robustness is limited. When the class of admissible preferences is expanded from disappointment aversion to the betweenness class, unique identification may fail.
\begin{prop}\label{prop:BP}
    There exists a random choice correspondence rationalized by more than one random betweenness preference.
\end{prop}
\begin{proof}
We provide a counterexample in the three-prize setting ($X=\{x_0,x_1,x_2\}$). Two random preferences $\mu$ and $\mu'$ are defined as follows. Under $\mu$, the realized preference is either $\succsim_1$ or $\succsim_2$ with equal probability. Figure~\ref{RIEU1} depicts their indifference maps. Under $\mu'$, the realized preference is either $\succsim'_1$ or $\succsim'_2$ with equal probability. Figure~\ref{RIEU2} depicts their indifference maps.
\par
These four preferences all belong to the class of betweenness preferences since their indifference sets are linear. They share the same indifference set containing $x_1$. Under $\succsim_1$, indifference sets intersect at a single point $r^1$ outside the simplex, whereas under $\succsim_2$, they intersect at a different point $r^2$. We then construct $\succsim_1'$ and $\succsim_2'$ by swapping the portions of the indifference maps of $\succsim_1$ and $\succsim_2$ that lie below the common indifference set containing $x_1$. 
\par
We claim that $\mu$ and $\mu'$ rationalize the same RCC. For example, consider lotteries $p,q$ shown in the figures. Since $p\succ_1q$ but $q\succ_2p$, each lottery is chosen with probability $1/2$ under $\mu$. The same is true under $\mu'$. More generally, by construction, the distribution of the indifference set containing any lottery $q$ is identical under $\mu$ and $\mu'$. It follows that $\mu$ and $\mu'$ assign the same choice probability to every lottery in every menu. Therefore, they rationalize the same RCC.
\end{proof}

\begin{figure}
\centering
\begin{tikzpicture}[line cap=round,line join=round,>=triangle 45,x=3cm,y=3.5cm]
\begin{scope}
\draw [line width=2pt,color=darkgray] (0,0)-- (1,0);
\draw [line width=2pt,color=darkgray] (1,0)-- (0,1);
\draw [line width=2pt,color=darkgray] (0,1)-- (0,0);
\draw [line width=0.7pt,dashed,color=red] (1.2,-0.1)--(0,0.5)--(-1,1);
\draw [line width=0.7pt,dashed,color=red] (1.2,-0.1)--(0,0)--(-0.3,1/40);
\draw [line width=0.7pt,dashed,color=red] (1.2,-0.1)--(0,0.25)--(-0.4,11/30);
\draw [line width=0.7pt,dashed,color=red] (1.2,-0.1)--(0,0.75)--(-0.4,31/30);
\draw [line width=0.7pt,dashed,color=red] (1.2,-0.1)--(0,1)--(-0.2,71/60);
\draw [line width=1pt,color=red] (1,0)-- (0,0.5);
\draw [line width=1pt,color=red] (6/7,0)-- (0,0.25);
\draw [line width=1pt,color=red] (6/7,1/7)-- (0,0.75);
\draw (-1,1.2) node[anchor=center] {(a) $\succsim_1$:};
\begin{scriptsize}
\draw [fill=blue] (1,0) circle (2pt);
\draw[color=blue] (1.1,0.07) node {$x_1$};
\draw [fill=blue] (0,0) circle (2pt);
\draw[color=blue] (-0.05,-0.05) node {$x_0$};
\draw [fill=blue] (0,1) circle (2pt);
\draw[color=blue] (0,1.07) node {$x_2$};
\draw [fill=gray] (1.2,-0.1) circle (2pt);
\draw[color=darkgray] (1.25,-0.15) node {$r^1$};
\draw [fill=gray] (-1,1) circle (2pt);
\draw[color=darkgray] (-1,0.9) node {$r^2$};
\draw [fill=blue] (0,0.25) circle (2pt);
\draw[color=blue] (0.06,0.3) node {$q'$};
\draw [fill=blue] (0.5,0) circle (2pt);
\draw[color=blue] (0.47,0.06) node {$p'$};
\draw [fill=blue] (0,0.75) circle (2pt);
\draw[color=blue] (0.06,0.8) node {$q$};
\draw [fill=blue] (0.5,0.5) circle (2pt);
\draw[color=blue] (0.56,0.56) node {$p$};
\end{scriptsize}
\end{scope}
\begin{scope}[xshift=6cm]
\draw [line width=2pt,color=darkgray] (0,0)-- (1,0);
\draw [line width=2pt,color=darkgray] (1,0)-- (0,1);
\draw [line width=2pt,color=darkgray] (0,1)-- (0,0);
\draw [line width=0.7pt,dashed,color=red] (1.2,-0.1)--(0,0.5)--(-1,1);
\draw [line width=0.7pt,dashed,color=red] (-1,1)--(0,0)--(0.2,-0.2);
\draw [line width=0.7pt,dashed,color=red] (-1,1)--(0,0.25)--(0.6,-0.2);
\draw [line width=0.7pt,dashed,color=red] (-1,1)--(0,0.75)--(0.75,9/16);
\draw [line width=0.7pt,dashed,color=red] (-1,1)--(0,1)--(0.5,1);
\draw [line width=1pt,color=red] (1,0)-- (0,0.5);
\draw [line width=1pt,color=red] (1/3,0)-- (0,0.25);
\draw [line width=1pt,color=red] (1/3,2/3)-- (0,0.75);
\draw (-1,1.2) node[anchor=center] {(b) $\succsim_2$:};
\begin{scriptsize}
\draw [fill=blue] (1,0) circle (2pt);
\draw[color=blue] (1.1,0.07) node {$x_1$};
\draw [fill=blue] (0,0) circle (2pt);
\draw[color=blue] (-0.05,-0.05) node {$x_0$};
\draw [fill=blue] (0,1) circle (2pt);
\draw[color=blue] (0,1.07) node {$x_2$};
\draw [fill=gray] (1.2,-0.1) circle (2pt);
\draw[color=darkgray] (1.25,-0.15) node {$r^1$};
\draw [fill=gray] (-1,1) circle (2pt);
\draw[color=darkgray] (-1,0.9) node {$r^2$};
\draw [fill=blue] (0,0.25) circle (2pt);
\draw[color=blue] (0.06,0.3) node {$q'$};
\draw [fill=blue] (0.5,0) circle (2pt);
\draw[color=blue] (0.47,0.06) node {$p'$};
\draw [fill=blue] (0,0.75) circle (2pt);
\draw[color=blue] (0.06,0.8) node {$q$};
\draw [fill=blue] (0.5,0.5) circle (2pt);
\draw[color=blue] (0.56,0.56) node {$p$};
\end{scriptsize}
\end{scope}
\end{tikzpicture}
\caption{}
\medskip 
\begin{minipage}{0.9\textwidth}
{\scriptsize (a): All indifference sets are linear. When extended to straight lines, they intersect at a single point $r^1$.\\(b): All indifference sets are linear. When extended to straight lines, they intersect at a single point $r^2$.\par}
\end{minipage}
\label{RIEU1}
\par\vspace{1cm}
\centering
\begin{tikzpicture}[line cap=round,line join=round,>=triangle 45,x=3cm,y=3.5cm]
\begin{scope}
\draw [line width=2pt,color=darkgray] (0,0)-- (1,0);
\draw [line width=2pt,color=darkgray] (1,0)-- (0,1);
\draw [line width=2pt,color=darkgray] (0,1)-- (0,0);
\draw [line width=0.7pt,dashed,color=red] (1.2,-0.1)--(0,0.5)--(-1,1);
\draw [line width=0.7pt,dashed,color=red] (-1,1)--(0,0)--(0.2,-0.2);
\draw [line width=0.7pt,dashed,color=red] (-1,1)--(0,0.25)--(0.6,-0.2);
\draw [line width=0.7pt,dashed,color=red] (1.2,-0.1)--(0,0.75)--(-0.4,31/30);
\draw [line width=0.7pt,dashed,color=red] (1.2,-0.1)--(0,1)--(-0.2,71/60);
\draw [line width=1pt,color=red] (1,0)-- (0,0.5);
\draw [line width=1pt,color=red] (1/3,0)-- (0,0.25);
\draw [line width=1pt,color=red] (6/7,1/7)-- (0,0.75);
\draw (-1,1.2) node[anchor=center] {(a) $\succsim_1'$:};
\begin{scriptsize}
\draw [fill=blue] (1,0) circle (2pt);
\draw[color=blue] (1.1,0.07) node {$x_1$};
\draw [fill=blue] (0,0) circle (2pt);
\draw[color=blue] (-0.05,-0.05) node {$x_0$};
\draw [fill=blue] (0,1) circle (2pt);
\draw[color=blue] (0,1.07) node {$x_2$};
\draw [fill=gray] (1.2,-0.1) circle (2pt);
\draw[color=darkgray] (1.25,-0.15) node {$r^1$};
\draw [fill=gray] (-1,1) circle (2pt);
\draw[color=darkgray] (-1,0.9) node {$r^2$};
\draw [fill=blue] (0,0.25) circle (2pt);
\draw[color=blue] (0.06,0.3) node {$q'$};
\draw [fill=blue] (0.5,0) circle (2pt);
\draw[color=blue] (0.47,0.06) node {$p'$};
\draw [fill=blue] (0,0.75) circle (2pt);
\draw[color=blue] (0.06,0.8) node {$q$};
\draw [fill=blue] (0.5,0.5) circle (2pt);
\draw[color=blue] (0.56,0.56) node {$p$};
\end{scriptsize}
\end{scope}
\begin{scope}[xshift=6cm]
\draw [line width=2pt,color=darkgray] (0,0)-- (1,0);
\draw [line width=2pt,color=darkgray] (1,0)-- (0,1);
\draw [line width=2pt,color=darkgray] (0,1)-- (0,0);
\draw [line width=0.7pt,dashed,color=red] (1.2,-0.1)--(0,0.5)--(-1,1);
\draw [line width=0.7pt,dashed,color=red] (1.2,-0.1)--(0,0)--(-0.3,1/40);
\draw [line width=0.7pt,dashed,color=red] (1.2,-0.1)--(0,0.25)--(-0.4,11/30);
\draw [line width=0.7pt,dashed,color=red] (-1,1)--(0,0.75)--(0.75,9/16);
\draw [line width=0.7pt,dashed,color=red] (-1,1)--(0,1)--(0.5,1);
\draw [line width=1pt,color=red] (1,0)-- (0,0.5);
\draw [line width=1pt,color=red] (6/7,0)-- (0,0.25);
\draw [line width=1pt,color=red] (1/3,2/3)-- (0,0.75);
\draw (-1,1.2) node[anchor=center] {(b) $\succsim_2'$:};
\begin{scriptsize}
\draw [fill=blue] (1,0) circle (2pt);
\draw[color=blue] (1.1,0.07) node {$x_1$};
\draw [fill=blue] (0,0) circle (2pt);
\draw[color=blue] (-0.05,-0.05) node {$x_0$};
\draw [fill=blue] (0,1) circle (2pt);
\draw[color=blue] (0,1.07) node {$x_2$};
\draw [fill=gray] (1.2,-0.1) circle (2pt);
\draw[color=darkgray] (1.25,-0.15) node {$r^1$};
\draw [fill=gray] (-1,1) circle (2pt);
\draw[color=darkgray] (-1,0.9) node {$r^2$};
\draw [fill=blue] (0,0.25) circle (2pt);
\draw[color=blue] (0.06,0.3) node {$q'$};
\draw [fill=blue] (0.5,0) circle (2pt);
\draw[color=blue] (0.47,0.06) node {$p'$};
\draw [fill=blue] (0,0.75) circle (2pt);
\draw[color=blue] (0.06,0.8) node {$q$};
\draw [fill=blue] (0.5,0.5) circle (2pt);
\draw[color=blue] (0.56,0.56) node {$p$};
\end{scriptsize}
\end{scope}
\end{tikzpicture}
\caption{}
\medskip
\begin{minipage}{0.9\textwidth}
{\scriptsize (a): All indifference sets are linear. When extended to straight lines, the indifference sets above $x_1$ intersect at $r^1$, whereas those below $x_1$ intersect at $r^2$.\\(b): All indifference sets are linear. When extended to straight lines, the indifference sets above $x_1$ intersect at $r^2$, whereas those below $x_1$ intersect at $r^1$.\par}
\end{minipage}
\label{RIEU2}
\end{figure}

This counterexample is analogous to the classic example of non-uniqueness of random preference in an abstract choice setting due to \cite{Fishburn1998}. In his construction, two new preferences are obtained by exchanging the lower portions of two existing rankings while keeping the upper portions fixed. Our construction can be viewed as a geometric analogue of Fishburn’s example.
\par
This non-uniqueness reflects a limitation of random choice under betweenness preferences: it does not reveal the joint distribution of choice across multiple menus. For instance, consider lotteries $p$, $q$, $p'$ and $q'$ depicted in Figures~\ref{RIEU1} and \ref{RIEU2}. Note that $p\succ_1(\succ'_1)q$ and $p'\prec_1(\succ'_1)q'$, and that $p\prec_2(\prec'_2)q$ and $p'\succ_2(\prec'_2)q'$. Therefore, under $\mu$, the probability that $p$ is chosen over $q$ {\it and} $p'$ is chosen over $q'$ is 0. However, under $\mu'$, this probability equals $\frac{1}{2}$. Thus, $\mu$ and $\mu'$ induce different joint distributions of choices across $\{p,q\}$ and $\{p',q'\}$, despite inducing the same RCC.
\par
Moreover, preferences $\succsim_1$ and $\succsim_2$ in Figure~\ref{RIEU1} are WEU preferences, whose defining geometric feature is that all indifference sets intersect at a common point outside the simplex. Preferences $\succsim_1'$ and $\succsim_2'$ in Figure~\ref{RIEU2}, on the other hand, belong to the more general class of semi-weighted utility \citep{Chew:1989}.\footnote{In the three-prize setting, semi-weighted utility preferences are characterized geometrically by two points outside the simplex: all indifference sets above the middle prize $x_1$ intersect at one point, whereas all indifference sets below $x_1$ intersect at another.} Therefore, the uniqueness of random preference fails even within the class of semi-weighted utility preferences.
\par
On the other hand, if we restrict preferences to the WEU class, the counterexample for \Cref{prop:BP} would not arise, since the preferences in Figure~\ref{RIEU2} are no longer admissible. This suggests the conjecture that unique identification may hold under WEU, especially since WEU preferences also admit a finite-dimensional parameterization and possess a strong geometric structure. Nevertheless, the next result shows that this conjecture is only partially correct. When there are only three prizes, as in the above counterexample, uniqueness is restored. Yet, when more prizes are introduced, non-uniqueness reappears. 
\begin{prop}\label{prop:WEU}
    In the three-prize setting ($|X|=3$), a random choice correspondence is rationalized by at most one random weighted expected utility preference. However, when $|X|\geq4$, there exists a random choice correspondence rationalized by more than one random weighted expected utility preference.
\end{prop}
\begin{proof}[Sketch of Proof]
    Under WEU, all indifference sets are not only hyperplanes but also share a common intersection outside the simplex. Hence, a preference can be determined by two hyperplanes passing through the best prize $x_N$ and the worst prize $x_0$, respectively, since their intersection determines the common focal point and hence all indifference sets. In a three-prize setting, this geometric structure yields a two-dimensional parameterization $(a_1,b_1)$, where $(a_1,1)$ and $(b_1,1)$ are the corresponding normal vectors. Moreover, the normal vector associated with $\lambda x_2+(1-\lambda)x_0$ is given by $(\lambda a_1+(1-\lambda)b_1,1)$. Therefore, random choice identifies the distribution of $\lambda a_1+(1-\lambda)b_1$ for every $\lambda\in[0,1]$. Since preferences are assumed to be monotone in stochastic dominance, $a_1$ and $b_1$ are supported on $[0,1]$. A Laplace-transform argument then shows that the joint distribution of $a_1$ and $b_1$ is uniquely determined.
    \par
    However, this proof strategy does not extend to the four-prize setting. With four prizes, the analogous geometric parameterization is four-dimensional, with $(a_1,a_2,1)$ and $(b_1,b_2,1)$ as the corresponding normal vectors. Nevertheless, unlike in the three-prize setting, random choice no longer identifies the distribution of every convex combination of these four parameters. Hence, a Laplace-transform argument breaks down. Instead, we show that there exist nontrivial perturbations in the parameter space that leave the distributions of normal vectors unchanged. Specifically, by introducing the differential operator $$L=\partial_{b_1}\partial_{d_2}-\partial_{b_2}\partial_{d_1}$$
    where $d_i=a_i-b_i$, and suitably choosing a smooth function $\Psi$, we construct two distinct density functions, $f$ and $f+\epsilon L\left[L[\Psi]\right]$, with $\epsilon>0$ sufficiently small, that induce the same RCC.
\end{proof}

\section{Concluding Remarks}
This paper shows that the uniqueness property of random expected utility extends beyond expected utility, though not universally. Unique identification persists under the disappointment aversion model. It fails, however, when risk preferences are further relaxed to the betweenness class. Even under weighted expected utility, another notable subclass of betweenness, non-uniqueness obtains with four or more prizes. Our findings raise the question of what features of a preference model give rise to unique identification of random utility.
\par
One might conjecture that our positive results arise because both disappointment aversion and weighted expected utility in a three-prize setting are only one parameter richer than expected utility. Nevertheless, our proof of the disappointment aversion result uses only the distributions of normal vectors associated with lotteries over the best and second-best prizes, leaving much of the available random choice data unused. With additional parameters, information from other lotteries may also be exploited for identification, so their being one-parameter extensions of expected utility does not fully explain our uniqueness results. This perspective is further supported by \cite{Lu:2020}, who establishes unique identification for random maxmin expected utility, even though maxmin preferences are parameterized by an infinite-dimensional set of beliefs. These observations suggest that parameter dimension alone is unlikely to be decisive, and that the structural properties underlying unique identification deserve further investigation.

\appendix
\appendixpage
\section{Proofs}
\subsection{Proof of Theorem \ref{thm:DA}}
Consider lottery $p=ax_N+(1-a)x_{N-1}$ for some $a\in[0,1]$. Under a DA representation with normalized $u$, the lottery value $V(p)$ must lie in $[u_{N-1},1]$. So it solves
$$V(p)=\sum_{x\in X}\phi(x,V(p))p(x)=u_{N-1}(1-a)+\frac{1+\beta V(p)}{1+\beta}\times a.$$
Thus, 
$$V(p)=\frac{a+(1+\beta)(1-a)u_{N-1}}{1+\beta(1-a)}.$$
Then, the normal vector of the indifference set containing $p$ is given by
\begin{align*}
    &\left(\phi(x_1,V(p)),\cdots,\phi(x_N,V(p))\right)=\left(u_1,\cdots,u_{N-1},\frac{1+\beta(1-a)u_{N-1}}{1+\beta(1-a)}\right)\\
    &\propto\left(\frac{u_1[1+\beta(1-a)]}{1+\beta(1-a)u_{N-1}},\cdots,\frac{u_{N-1}[1+\beta(1-a)]}{1+\beta(1-a)u_{N-1}},1\right)\coloneqq(Y_{1,a},\cdots,Y_{N-1,a},1).
\end{align*}
\par
Let $\mu$ be the joint distribution of $(u_1,\cdots,u_{N-1},\beta)$, corresponding to the random DA preference that rationalizes $\rho$, where $0<u_1<\cdots<u_{N-1}<1$ and $\beta>-1$. By the proof of Theorem \ref{thm:EU} and the subsequent discussion, random choice identifies the distribution of the normal vector $(Y_{1,a},\cdots,Y_{N-1,a},1)$ for every $a\in[0,1]$.
\par
Define the following transformations:
$$Z_k\coloneqq\frac{u_k}{u_{N-1}}\ \forall\ k=1,\cdots,N-2;\ Z_{N-1}\coloneqq\frac{u_{N-1}}{1-u_{N-1}};\ Z_{N}\coloneqq \frac{(1+\beta)u_{N-1}}{1-u_{N-1}}.$$
Since $0<u_1<\cdots<u_{N-1}<1$ and $\beta>-1$, $Z_1,\cdots,Z_N$ are all positive.
\par
Observe that 
$$Z_k=\frac{Y_{k,a}}{Y_{N-1,a}}\ \forall\ k=1,\cdots,N-2,$$
and
$$\frac{Y_{N-1,a}}{1-Y_{N-1,a}}=Z_{N-1}(1+\beta(1-a))=aZ_{N-1}+(1-a)Z_{N}.$$
Hence, for each $a\in[0,1]$, the joint distribution of
$$\left(Z_1,\cdots,Z_{N-2},aZ_{N-1}+(1-a)Z_{N}\right)$$
is identified. 
\par
Now, we can compute
$$\mathbb{E}\left[\prod_{k=1}^{N-2}e^{-r_kZ_k}e^{-r_{N-1}(aZ_{N-1}+(1-a)Z_{N})}\right]$$
for all $r_1,\cdots,r_{N-1}\geq0$. Thus, for any nonnegative integers $m_1,\cdots,m_{N}$ with $m_{N-1}+m_{N}>0$, by taking $r_k=m_k$ for all $k=1,\cdots,N-2$, $r_{N-1}=m_{N-1}+m_{N}$, and $a=\frac{m_{N-1}}{m_{N-1}+m_{N}}$, we obtain 
$$\mathbb{E}\left[\prod_{k=1}^{N}e^{-m_kZ_k}\right].$$
If $m_{N-1}=m_N=0$, the same expression is obtained by taking $r_{N-1}=0$.
\par
Since $\left(e^{-Z_1},\cdots,e^{-Z_{N}}\right)$ has support contained in $[0,1]^{N}$, and since we have identified all of its joint moments, the distribution of $\left(e^{-Z_1},\cdots,e^{-Z_{N}}\right)$ is uniquely determined. Since the logarithm is one-to-one, we obtain the distribution of $(Z_1,\cdots,Z_{N})$.
\par
Finally, the variables $(u_1,\ldots,u_{N-1},\beta)$ can be recovered from $(Z_1,\cdots,Z_{N})$ through 
$$u_{k}=\frac{Z_kZ_{N-1}}{1+Z_{N-1}}\ \forall\ k=1,\cdots,N-2,\ u_{N-1}=\frac{Z_{N-1}}{1+Z_{N-1}},\ \text{and}\ \beta=\frac{Z_{N}}{Z_{N-1}}-1.$$ We therefore uniquely identify the distribution $\mu$ of $(u_1,\cdots,u_{N-1},\beta)$. Hence, $\rho$ cannot be rationalized by more than one random DA preference.


\subsection{Proof of Proposition \ref{prop:WEU}}
For any weighted expected utility representation
$$V(p)=\frac{\sum_{n=0}^Nu_nw_np_n}{\sum_{n=0}^Nw_np_n}$$
with normalization $u_0=0$ and $u_N=w_0=w_N=1$, the indifference set at utility level $v\in[0,1]$ is characterized by
\begin{align*}
    &\frac{\sum_{n=0}^Nu_nw_np_n}{\sum_{n=0}^Nw_np_n}=v\iff \sum_{n=0}^Nw_n\left(u_n-v\right)p_n=0\\
    &\iff\sum_{n=1}^N\left[w_nu_n+(1-w_n)v\right]p_n=v.
\end{align*}
Let $K(p)=(w_1u_1+(1-w_1)V(p),\cdots,w_{N-1}u_{N-1}+(1-w_{N-1})V(p),1)\in\R^N$. That is, $K(p)$ is the normal vector of the indifference set containing lottery $p$.

\paragraph{Uniqueness in the three-prize setting:}
\par\ \par
Now, suppose that $X=\{x_0,x_1,x_2\}$. Consider lotteries supported on $\{x_0,x_2\}$. We have for all $\lambda\in[0,1]$,
$$V(\lambda x_2+(1-\lambda)x_0)=\frac{u_0w_0(1-\lambda)+u_2w_2\lambda}{w_0(1-\lambda)+w_2\lambda}=\lambda.$$
It follows that
\begin{align*}
    &K(\lambda x_2+(1-\lambda)x_0)=(w_1u_1+(1-w_1)\lambda,1)\\
    &=\lambda(w_1u_1+1-w_1,1)+(1-\lambda)(w_1u_1,1)=\lambda K(x_2)+(1-\lambda)K(x_0).
\end{align*}
That is, the normal vector associated with $\lambda x_2+(1-\lambda)x_0$ is exactly the $\lambda$-weighted average of the normal vectors associated with $x_2$ and $x_0$ respectively. 
\par
Let $a_1=w_1u_1+1-w_1$ and $b_1=w_1u_1$. Monotonicity with respect to first-order stochastic dominance implies $a_1,b_1\in[0,1]$.  
\par
By the proof of Theorem \ref{thm:EU} and the subsequent discussion, random choice identifies the distribution of $K(\lambda x_2+(1-\lambda)x_0)$ for every $\lambda\in[0,1]$. Thus, we have the distribution of $\lambda a_1+(1-\lambda)b_1$ for all $\lambda\in[0,1]$. Now, we can compute $\E[\exp\{-r(\lambda a_1+(1-\lambda)b_1)\}]$ for all $\lambda\in[0,1]$ and $r\geq 0$. For any $s,t\geq 0$ with $s+t>0$, take $r=s+t$ and $\lambda=\frac{s}{s+t}$. We then recover $\E[\exp\{-sa_1-tb_1\}]$ for all $s,t\geq0$, i.e.,\ the Laplace transform of $(a_1,b_1)$. Since $(a_1,b_1)$ is supported on $[0,1]^2$, the Laplace transform uniquely determines its distribution. 
\par
Since $w_1=1-a_1+b_1$ and $u_1=\frac{b_1}{1-a_1+b_1}$, we can recover the distribution of $(u_1,w_1)$ uniquely. Therefore, in the three-prize setting, an RCC cannot be rationalized by more than one random WEU preference. 

\paragraph{Example of non-uniqueness with four or more prizes:}
\par\ \par
Suppose that $X=\{x_0,x_1,x_2,x_3\}$. We have
$$K(x_3)=(w_1u_1+1-w_1,w_2u_2+1-w_2,1)$$
and
$$K(x_0)=(w_1u_1,w_2u_2,1).$$
For all $n=1,2$, let 
$$b_n=w_nu_n\ \text{and}\ d_n=1-w_n.$$
Then the normal vector $K(p)$ can be expressed in terms of $b\coloneqq(b_1,b_2,1)$ and $d\coloneqq(d_1,d_2,0)$. Assume that $p\sim\lambda x_3+(1-\lambda)x_0$. Then
$$p\cdot K(p)=(\lambda x_3+(1-\lambda)x_0)\cdot K(p)=\lambda.$$
Moreover, $K(p)=K(\lambda x_3+(1-\lambda)x_0)=b+\lambda d$. Thus, $K(p)$ satisfies
$$K(p)=b+(p\cdot K(p))d.$$
Multiplying both sides by $p$ yields
$$p\cdot K(p)=p\cdot b+(p\cdot K(p))\times(p\cdot d).$$
Solving for $p\cdot K(p)$ from this equation, we have
$$\lambda=p\cdot K(p)=\frac{p\cdot b}{1-p\cdot d}.$$
It follows that
$$K(p)=b+\frac{p\cdot b}{1-p\cdot d}\times d.$$
\par
Now, we proceed to construct an example of non-uniqueness. Since $w_n=1-d_n$ and $u_n=\frac{b_n}{1-d_n}$, a random WEU preference can be represented by a distribution of $(b_1,b_2,d_1,d_2)$. Since preferences are monotone in first-order stochastic dominance, the parameter space is
$$\Omega\coloneqq\{(b_1,b_2,d_1,d_2):0<b_1<b_2<1,\ 0<b_1+d_1<b_2+d_2<1\}\subseteq\mathbb{R}^4.$$
Let $\lambda_{p,b,d}=\frac{p\cdot b}{1-p\cdot d}$. Since the last coordinate of $K(p)$ is always equal to 1, we focus on its first two coordinates and let $K(p|b,d)=(b_1+\lambda_{p,b,d}d_1,b_2+\lambda_{p,b,d}d_2)$.
\par
Define the differential operator:
$$L\coloneqq\partial_{b_1}\partial_{d_2}-\partial_{b_2}\partial_{d_1}.$$
Then
\begin{align*}
&L[b_1+\lambda_{p,b,d}d_1]=\frac{\partial}{\partial d_2}\left[1+\frac{d_1p_1}{1-p\cdot d}\right]-\frac{\partial}{\partial d_1}\left[\frac{d_1p_2}{1-p\cdot d}\right]\\
&=\frac{d_1p_1p_2}{(1-p\cdot d)^2}-\frac{p_2(1-p\cdot d)-(-p_1)d_1p_2}{(1-p\cdot d)^2}=-\frac{p_2}{1-p\cdot d}.
\end{align*}
and
\begin{align*}
&L[b_2+\lambda_{p,b,d}d_2]=\frac{\partial}{\partial d_2}\left[\frac{d_2p_1}{1-p\cdot d}\right]-\frac{\partial}{\partial d_1}\left[1+\frac{d_2p_2}{1-p\cdot d}\right]\\
&=\frac{p_1(1-p\cdot d)-(-p_2)d_2p_1}{(1-p\cdot d)^2}-\frac{d_2p_1p_2}{(1-p\cdot d)^2}=\frac{p_1}{1-p\cdot d}.
\end{align*}
So
$$L[K(p|b,d)]=\left(-\frac{p_2}{1-p\cdot d},\frac{p_1}{1-p\cdot d}\right).$$
Moreover,
$$\frac{\partial K(p|b,d)}{\partial d_1}=(p\cdot b)\times\left(\frac{1-p\cdot d+d_1p_1}{(1-p\cdot d)^2},\frac{d_2p_1}{(1-p\cdot d)^2}\right)=\lambda_{p,b,d}\times\frac{\partial K(p|b,d)}{\partial b_1}$$
and
$$\frac{\partial K(p|b,d)}{\partial d_2}=(p\cdot b)\times\left(\frac{d_1p_2}{(1-p\cdot d)^2},\frac{1-p\cdot d+d_2p_2}{(1-p\cdot d)^2}\right)=\lambda_{p,b,d}\times\frac{\partial K(p|b,d)}{\partial b_2}.$$
Now, take $\theta=(\theta_1,\theta_2)\in\mathbb{R}^2$. For notational convenience, write $K$ and $\lambda$ for $K(p|b,d)$ and $\lambda_{p,b,d}$, respectively. We have
\begin{align*}
    &L[\exp(\theta\cdot K)]=\frac{\partial}{\partial b_1}\left[e^{\theta\cdot K}\theta\cdot\frac{\partial K}{\partial d_2}\right]-\frac{\partial}{\partial b_2}\left[e^{\theta\cdot K}\theta\cdot\frac{\partial K}{\partial d_1}\right]\\
    =&e^{\theta\cdot K}\left\{\theta\cdot\frac{\partial^2K}{\partial b_1\partial d_2}+\left(\theta\cdot\frac{\partial K}{\partial b_1}\right)\left(\theta\cdot\frac{\partial K}{\partial d_2}\right)-\theta\cdot\frac{\partial^2K}{\partial b_2\partial d_1}-\left(\theta\cdot\frac{\partial K}{\partial b_2}\right)\left(\theta\cdot\frac{\partial K}{\partial d_1}\right)\right\}\\
    =&e^{\theta\cdot K}\left\{\theta\cdot L[K]+\lambda\left(\theta\cdot\frac{\partial K}{\partial b_1}\right)\left(\theta\cdot\frac{\partial K}{\partial b_2}\right)-\lambda\left(\theta\cdot\frac{\partial K}{\partial b_2}\right)\left(\theta\cdot\frac{\partial K}{\partial b_1}\right)\right\}\\
    =&e^{\theta\cdot K}\theta\cdot L[K]=e^{\theta\cdot K}\times\frac{\theta_2p_1-\theta_1p_2}{1-p\cdot d}.
\end{align*}
Note that
\begin{align*}
    &L\left[\frac{e^{\theta\cdot K}}{1-p\cdot d}\right]\\
    =&\frac{\partial}{\partial b_1}\left[\frac{e^{\theta\cdot K}\theta\cdot\frac{\partial K}{\partial d_2}(1-p\cdot d)+p_2e^{\theta\cdot K}}{(1-p\cdot d)^2}\right]-\frac{\partial}{\partial b_2}\left[\frac{e^{\theta\cdot K}\theta\cdot\frac{\partial K}{\partial d_1}(1-p\cdot d)+p_1e^{\theta\cdot K}}{(1-p\cdot d)^2}\right]\\
    =&\frac{L[\exp(\theta\cdot K)]}{1-p\cdot d}+\frac{p_2e^{\theta\cdot K}\theta\cdot\frac{\partial K}{\partial b_1}}{(1-p\cdot d)^2}-\frac{p_1e^{\theta\cdot K}\theta\cdot\frac{\partial K}{\partial b_2}}{(1-p\cdot d)^2}\\
    =&\frac{e^{\theta\cdot K}}{(1-p\cdot d)^2}\left\{\theta_2p_1-\theta_1p_2+\theta\cdot\left(\frac{p_2(1-p\cdot d)}{1-p\cdot d},\frac{-p_1(1-p\cdot d)}{1-p\cdot d}\right)\right\}\\
    =&\frac{e^{\theta\cdot K}}{(1-p\cdot d)^2}\times0=0.
\end{align*}
Hence, where $L^2$ denotes $L\circ L$, 
$$L^2[\exp(\theta\cdot K(p|b,d))]=L\left[e^{\theta\cdot K(p|b,d)}\times\frac{\theta_2p_1-\theta_1p_2}{1-p\cdot d}\right]=0.$$
\par
Let $\Psi$ be a real-valued function on $\mathbb{R}^4$ such that (i) $\Psi$ is infinitely differentiable, (ii) $\Psi$ has a compact support that is a subset of $\Omega$, and (iii) $L^2[\Psi]\not\equiv 0$. Such a function exists because $L^2$ is a non-zero differential operator and $\Omega$ has nonempty interior.
\par
Let $$h=L^2[\Psi].$$
Now, we argue that
$$\int_\Omega e^{\theta\cdot K(p|b,d)}h(b,d)\ \mathrm{d}(b,d)=0.$$
Note that
$$\int_\Omega e^{\theta\cdot K(p|b,d)}h(b,d)=\int_\Omega e^{\theta\cdot K(p|b,d)}\left(\partial_{b_1}\partial_{d_2}L[\Psi]-\partial_{b_2}\partial_{d_1}L[\Psi]\right).$$
Fixing $(b_2,d_1,d_2)$, integrating by parts yields
$$\int_\mathbb{R}e^{\theta\cdot K}\frac{\partial^2L[\Psi]}{\partial b_1\partial d_2}\mathrm{d}b_1=e^{\theta\cdot K}\frac{\partial L[\Psi]}{\partial d_2}\bigg|_{b_1\downarrow-\infty}^{b_1\uparrow\infty}-\int_\mathbb{R}\frac{\partial L[\Psi]}{\partial d_2}\frac{\partial e^{\theta\cdot K}}{\partial b_1}=-\int_\mathbb{R}\frac{\partial L[\Psi]}{\partial d_2}\frac{\partial e^{\theta\cdot K}}{\partial b_1},$$
where the boundary terms vanish since $\Psi$ has compact support contained in the interior of $\Omega$. Hence, by integrating by parts twice, we have
$$\int_\Omega e^{\theta\cdot K}\times\frac{\partial^2L[\Psi]}{\partial b_1\partial d_2}=-\int_\Omega\frac{\partial L[\Psi]}{\partial d_2}\frac{\partial e^{\theta\cdot K}}{\partial b_1}=\int_\Omega L[\Psi]\times\frac{\partial^2e^{\theta\cdot K}}{\partial b_1\partial d_2}.$$
Similarly,
$$\int_\Omega e^{\theta\cdot K}\times\frac{\partial^2L[\Psi]}{\partial b_2\partial d_1}=-\int_\Omega\frac{\partial L[\Psi]}{\partial d_1}\frac{\partial e^{\theta\cdot K}}{\partial b_2}=\int_\Omega L[\Psi]\times\frac{\partial^2e^{\theta\cdot K}}{\partial b_2\partial d_1}.$$
Thus,
\begin{align*}
&\int_\Omega e^{\theta\cdot K}\times\left(\frac{\partial^2 L[\Psi]}{\partial b_1\partial d_2}-\frac{\partial^2L[\Psi]}{\partial b_2\partial d_1}\right)=\int_\Omega L[\Psi]\times\left(\frac{\partial^2e^{\theta\cdot K}}{\partial b_1\partial d_2}-\frac{\partial^2e^{\theta\cdot K}}{\partial b_2\partial d_1}\right)\\
&=\int_\Omega L[e^{\theta\cdot K}]\left(\partial_{b_1}\partial_{d_2}\Psi-\partial_{b_2}\partial_{d_1}\Psi\right).
\end{align*}
By the same trick of integration by parts, we have
$$\int_\Omega L[e^{\theta\cdot K(p|b,d)}]\left(\partial_{b_1}\partial_{d_2}\Psi-\partial_{b_2}\partial_{d_1}\Psi\right)=\int_\Omega\Psi\times L^2[e^{\theta\cdot K(p|b,d)}].$$
Therefore,
$$\int_\Omega e^{\theta\cdot K(p|b,d)}h(b,d)=\int_\Omega\Psi\times L^2[e^{\theta\cdot K(p|b,d)}]=\int_\Omega\Psi\times0=0.$$
In particular, by taking $\theta=(0,0)$, we have $$\int_\Omega h=0.$$
\par
Let $f$ be a density function on $\Omega$ such that $f$ is bounded below by a positive number on the support of $\Psi$. Since $h=L^2[\Psi]\not\equiv0$, $f+\epsilon h\neq f$ whenever $\epsilon>0$, and $f+\epsilon h\geq 0$ on $\Omega$ for $\epsilon>0$ small enough. Since 
$$\int_\Omega(f+\epsilon h)=\int_\Omega f+\epsilon\int_\Omega h=1+0=1,$$ 
$f+\epsilon h$ is also a density function on $\Omega$. Moreover,
$$\int_\Omega e^{\theta\cdot K(p|b,d)}[f(b,d)+\epsilon h(b,d)]=\int_\Omega e^{\theta\cdot K(p|b,d)}f(b,d)+\epsilon\times0=\int_\Omega e^{\theta\cdot K(p|b,d)}f(b,d).$$
Thus, $f$ and $f+\epsilon h$ induce the same Laplace transform of $K(p|b,d)$. Since $K(p|b,d)\in [0,1]^2$, a distribution of $K(p|b,d)$ has a compact support. Hence, equality of Laplace transforms implies equality of distributions. In other words, for every lottery $p\in\Delta$, the two densities induce the same distribution of the normal vector of the indifference set containing $p$. Therefore, they induce the same RCC. 
\par
When $|X|>4$, we can extend the construction of this four-prize counterexample by keeping the utility and weight parameters of the additional prizes fixed. As before, write $b=(b_1,b_2,\cdots,b_{N-1},1)$ and $d=(d_1,d_2,\cdots,d_{N-1},0)$ with $b_n=w_nu_n$ and $d_n=1-w_n$ for all $n=1,\cdots,N-1$. Let $K_n(p|b,d)=b_n+\lambda_{p,b,d}d_n$. Now, $K(p|b,d)=(K_n(p|b,d))_{n=1}^{N-1}\in\R^{N-1}$, and the normal vector $K(p)=(K(p|b,d),1)$. Again, take the differential operator $L=\partial_{b_1}\partial_{d_2}-\partial_{b_2}\partial_{d_1}$.
\par
One can verify that 
$$\frac{\partial K_n(p|b,d)}{\partial d_i}=\lambda_{p,b,d}\times\frac{\partial K_n(p|b,d)}{\partial b_i}$$
for all $n=1,\cdots,N-1$ and $i=1,2$. Note that $L[\lambda_{p,b,d}]=0$. Thus, for $n\geq 3$, $L[K_n(p|b,d)]=d_nL[\lambda_{p,b,d}]=0$. Let $\theta\in\R^{N-1}$. Then, by the same calculation as before,
$$L[\exp(\theta\cdot K(p|b,d))]=e^{\theta\cdot K(p|b,d)}\theta\cdot L[K(p|b,d)]=e^{\theta\cdot K(p|b,d)}\times\frac{\theta_2p_1-\theta_1p_2}{1-p\cdot d}.$$
Moreover, one can also establish $L\left[\frac{e^{\theta\cdot K(p|b,d)}}{1-p\cdot d}\right]=0$ by a similar calculation as in the four-prize setting. It follows that $L^2[\exp(\theta\cdot K(p|b,d))]=0$.
\par
Fix $b_n,d_n$ for $n\geq 3$ (i.e., they are assumed non-random). Then a random WEU remains represented by a distribution of $(b_1,b_2,d_1,d_2)$ supported on
$$\{(b_1,b_2,d_1,d_2):0<b_1<b_2<b_3,\ 0<b_1+d_1<b_2+d_2<b_3+d_3\}\subseteq\mathbb{R}^4.$$
Pick $\Psi:\R^4\to\R$ as before and let $h=L^2[\Psi]$. The integration-by-parts argument then proceeds exactly as in the four-prize case. Thus, we can construct two distinct densities that induce the same RCC.

\linespread{1}\selectfont
\bibliographystyle{apalike}
\bibliography{RNEU}

@article{Lu:2020,
author = {Lu, Jay},
title = {Random ambiguity},
journal = {Theoretical Economics},
volume = {16},
number = {2},
pages = {539-570},
doi = {https://doi.org/10.3982/TE3810},
url = {https://onlinelibrary.wiley.com/doi/abs/10.3982/TE3810},
eprint = {https://onlinelibrary.wiley.com/doi/pdf/10.3982/TE3810},
year = {2021}
}

@article{Chew:1989,
	Author = {Soo Hong Chew},
	Journal = {Annals of Operations Research},
	Pages = {273-298},
	Title = {Axiomatic utility theories with the betweenness property},
	Volume = {19},
	Year = {1989}}

@article{Gul:2006,
	Author = {Faruk Gul and Wolfgang Pesendorfer},
	Journal = {Econometrica},
	Number = {1},
	Pages = {121-146},
	Title = {Random expected utility},
	Volume = {74},
	Year = {2006}}

@article{Chew:1983,
	Author = {Soo Hong Chew},
	Journal = {Econometrica},
	Number = {4},
	Pages = {1065-1092},
	Title = {A Generalization of the Quasilinear Mean with Applications to the Measurement of Income Inequality and Decision Theory Resolving the Allais Paradox},
	Volume = {51},
	Year = {1983}}

@article{Gul:1991,
	Author = {Faruk Gul},
	Journal = {Econometrica},
	Number = {3},
	Pages = {667-686},
	Title = {A Theory of Disappointment Aversion},
	Volume = {59},
	Year = {1991}}

@article{Dekel:1986,
	Author = {Eddie Dekel},
	Journal = {Journal of Economic Theory},
	Number = {2},
	Pages = {304-318},
	Title = {An axiomatic characterization of preferences under uncertainty: Weakening the independence axiom},
	Volume = {40},
	Year = {1986}}

@incollection{Fishburn1998,
  author    = {Fishburn, Peter C.},
  title     = {Stochastic Utility},
  editor     = {Barber\`a, Salvador and Hammond, Peter J. and Seidl, Christian},
  booktitle = {Handbook of Utility Theory, Volume 1: Principles},
  pages      = {273--319},
  year       = {1998},
  publisher  = {Kluwer Academic Publishers},
  address    = {Dordrecht}}

@article{SULEYMANOV2024,
title = {Branching-independent random utility model},
journal = {Journal of Economic Theory},
volume = {220},
pages = {105880},
year = {2024},
issn = {0022-0531},
doi = {https://doi.org/10.1016/j.jet.2024.105880},
url = {https://www.sciencedirect.com/science/article/pii/S0022053124000863},
author = {Elchin Suleymanov}
}

@article{TURANSICK2022,
title = {Identification in the random utility model},
journal = {Journal of Economic Theory},
volume = {203},
pages = {105489},
year = {2022},
issn = {0022-0531},
doi = {https://doi.org/10.1016/j.jet.2022.105489},
url = {https://www.sciencedirect.com/science/article/pii/S0022053122000795},
author = {Christopher Turansick}
}

@unpublished{DuttaMasatliogluMu2026,
  author = {Dutta, Rohan and Masatlioglu, Yusufcan and Mu, Xiaosheng},
  title = {Binary Choice Paths and Random Utility},
  year = {2026},
  note = {SSRN Working Paper},
  url = {https://ssrn.com/abstract=6177358}
}

@incollection{BlockMarschak1960,
  author    = {Block, H. David and Marschak, Jacob},
  title     = {Random Orderings and Stochastic Theories of Responses},
  booktitle = {Contributions to Probability and Statistics: Essays in Honor of Harold Hotelling},
  editor    = {Olkin, Ingram},
  pages     = {97--132},
  publisher = {Stanford University Press},
  year      = {1960}
}

@article{single-crossing-RUM,
 ISSN = {00129682, 14680262},
 URL = {http://www.jstor.org/stable/44955133},
 author = {Jose Apesteguia and Miguel A. Ballester and Jay Lu},
 journal = {Econometrica},
 number = {2},
 pages = {661--674},
 publisher = {[Wiley, The Econometric Society]},
 title = {SINGLE-CROSSING RANDOM UTILITY MODELS},
 urldate = {2026-06-23},
 volume = {85},
 year = {2017}
}

@article{YANG2023,
title = {Random quasi-linear utility},
journal = {Journal of Economic Theory},
volume = {209},
pages = {105650},
year = {2023},
issn = {0022-0531},
doi = {https://doi.org/10.1016/j.jet.2023.105650},
url = {https://www.sciencedirect.com/science/article/pii/S0022053123000467},
author = {Erya Yang and Igor Kopylov}
}

@article{Lu:2016,
author = {Lu, Jay},
title = {Random Choice and Private Information},
journal = {Econometrica},
volume = {84},
number = {6},
pages = {1983-2027},
doi = {https://doi.org/10.3982/ECTA12821},
url = {https://onlinelibrary.wiley.com/doi/abs/10.3982/ECTA12821},
eprint = {https://onlinelibrary.wiley.com/doi/pdf/10.3982/ECTA12821},
year = {2016}
}

@book{Billingsley1995,
  author    = {Patrick Billingsley},
  title     = {Probability and Measure},
  edition   = {3},
  publisher = {John Wiley \& Sons},
  year      = {1995}
}

\end{NoHyper}\end{document}